\documentclass[11pt,envcountsame]{article}

\usepackage[a4paper,margin=1in]{geometry}
\usepackage{xspace}
\usepackage{comment}
\usepackage{amsmath}
\usepackage{amsfonts}
\usepackage{amssymb}
\usepackage{amsthm}
\usepackage{thmtools,thm-restate}
\usepackage[textsize=small]{todonotes}
\usepackage{tikz}
\usetikzlibrary{arrows,automata}
\usepackage{stmaryrd}
\usepackage{mathdots}
\usepackage{algorithm}
\usepackage[noend]{algpseudocode}
\usepackage[bookmarks,bookmarksopen,bookmarksdepth=2]{hyperref}

\theoremstyle{definition}\newtheorem{theorem}{Theorem}
\theoremstyle{definition}\newtheorem{corollary}[theorem]{Corollary}
\theoremstyle{definition}\newtheorem{example}[theorem]{Example}
\theoremstyle{definition}\newtheorem{lemma}[theorem]{Lemma}
\theoremstyle{definition}\newtheorem{proposition}[theorem]{Proposition}
\theoremstyle{definition}\newtheorem{claim}[theorem]{Claim}

\title{Reachability in Vector Addition Systems is Ackermann-complete}
\author{Wojciech Czerwiński\thanks{This work was partially supported by the ERC grant INFSYS, agreement no. 950398} \\ University of Warsaw           \\ \texttt{wczerwin@mimuw.edu.pl} \and
Łukasz Orlikowski    \\ University of Warsaw           \\ \texttt{lo418363@students.mimuw.edu.pl}}

\newcommand{\ignore}[1]{}

\newcommand{\size}{\textsc{size}}

\newcommand{\F}{\mathcal{F}}

\newcommand{\A}{\mathcal{A}}

\newcommand{\N}{\mathbb{N}}
\newcommand{\Z}{\mathbb{Z}}

\newcommand{\reach}{\textsc{Reach}}

\newcommand{\trans}[1]{\stackrel{#1}{\longrightarrow}}

\newcommand{\trg}{\textup{trg}}
\newcommand{\src}{\textup{src}}

\newcommand{\conf}{\textup{Conf}}

\newcommand{\inp}{\textup{in}}
\newcommand{\out}{\textup{out}}
\newcommand{\reaches}{\longrightarrow}

\newcommand{\eff}{\textup{eff}}

\DeclareSymbolFont{extraup}{U}{zavm}{m}{n}
\DeclareMathSymbol{\varheart}{\mathalpha}{extraup}{86} 
\DeclareMathSymbol{\vardiamond}{\mathalpha}{extraup}{87} 

\newcommand{\kqhide}[1]{\textcolor{teal}{sth old is hidden here}}

\newcommand{\nl}{\textup{NL}\xspace}

\newcommand{\np}{\textup{NP}\xspace}
\newcommand{\pspace}{\textup{PSpace}\xspace}

\newcommand{\expspace}{\textup{ExpSpace}\xspace}

\newcommand{\tower}{\textup{Tower}\xspace}
\newcommand{\ackermann}{\textup{Ackermann}\xspace}

\newcommand{\zerotest}[1]{\textbf{zero-test}($#1$)}

\newcommand{\add}[2]{$\coreadd{#1}{#2}$}
\newcommand{\sub}[2]{$\coresub{#1}{#2}$}
\newcommand{\coreadd}[2]{#1 \,\, +\!\!= \, #2}
\newcommand{\coresub}[2]{#1 \,\, -\!\!= \, #2}
\newcommand{\vr}[1]{#1}

\begin{document}

\maketitle

\begin{abstract}
Vector Addition Systems and equivalent Petri nets are a well established models of concurrency.
The central algorithmic problem for Vector Addition Systems with a long research history is the reachability problem
asking whether there exists a run from one given configuration to another.
We settle its complexity to be Ackermann-complete thus closing the problem open for 45 years.
In particular we prove that the problem is $\F_k$-hard for Vector Addition Systems with States in dimension $6k$,
where $\F_k$ is the $k$-th complexity class from the hierarchy of fast-growing complexity classes.
\end{abstract}

\section{Introduction}
The model of Vector Addition Systems (VASes) is a fundamental computation model well suited to model concurrent phenomena.
Together with essentially equivalent Petri nets it is long studied and has numerous applications in modelling and analysis
of computer systems and natural processes. The central algorithmic problem for VASes is the reachability problem,
which asks whether there exists a run from one given configuration to another.
The reachability problem has a long research history.
In 1976 it was shown to be \expspace-hard by Lipton~\cite{Lipton76}.
Decidability of the reachability problem was first proven by Mayr in 1981~\cite{DBLP:conf/stoc/Mayr81}.
The construction was simplified later by Kosaraju~\cite{DBLP:conf/stoc/Kosaraju82} and Lambert~\cite{DBLP:journals/tcs/Lambert92}.
Their approach was to use an equivalent model of VASes with states (VASSes) and in certain situations, when the answer to the problem
is not clear use a nontrivial decomposition of the system into simpler ones.
This technique is called the KLM decomposition after the names of its three inventors.
Despite a substantial effort of the community for a long time there was no known upper complexity bound for the VASS reachability problem.
There were however important results in the special cases when the dimension is fixed.
Haase et al. shown an \np-completeness of the problem in binary encoded one-dimensional VASSes~\cite{DBLP:conf/concur/HaaseKOW09}.
In dimension one the reachability problem for unary encoded VASSes can be easily shown \nl-complete.
In 2015 Blondin et al. have shown that the reachability problem for two-dimensional VASSes is \pspace-complete in the case
when transitions are encoded in binary~\cite{DBLP:conf/lics/BlondinFGHM15}. Further improvement came soon after that,
a year later Englert et al. proved that the same problem in the case of unary encodings of transitions is \nl-complete~\cite{DBLP:conf/lics/EnglertLT16}.

In 2015 Leroux and Schmitz have obtained the first upper complexity bound for the reachability problem proving that it belongs
to the cubic-Ackermannian complexity class denoted also $\F_{\omega^3}$~\cite{DBLP:conf/lics/LerouxS15}.
The same authors have improved their result recently in 2019 showing that the problem can be solved in the Ackermann complexity class
(denoted $\F_\omega$)~\cite{DBLP:conf/lics/LerouxS19}.
They have actually shown that the reachability problem for $k$-dimensional VASSes
(denoted $k$-VASSes) can be solved in the complexity class $\F_{k+4}$, where $\F_i$ is the hierarchy of complexity classes
related to the hierarchy of fast-growing functions $F_i$.
In the meanwhile in~\cite{DBLP:conf/stoc/CzerwinskiLLLM19} it was shown that the reachability problem is \tower-hard,
recall that $\tower = \F_3$.
Thus the complexity gap was decreased to the gap between \tower and \ackermann complexity classes.

\paragraph*{Our contribution}
In this paper we close the above mentioned complexity gap.
Our main result is actually a more detailed hardness result, which depends on the dimension of the input VASS.

\begin{theorem}\label{thm:main}
For each $k \geq 3$ the reachability problem for $6k$-VASSes is $\F_k$-hard.
\end{theorem}

In particular the reachability problem for $18$-VASSes is $\tower$-hard, as $\tower = \F_3$.
An immediate consequence of Theorem~\ref{thm:main} is that reachability problem for VASSes is \ackermann-hard.
Together with~\cite{DBLP:conf/lics/LerouxS19} it implies the following.

\begin{corollary}
The VASS reachability problem is \ackermann-complete.
\end{corollary}

Recently J{\'{e}}r{\^{o}}me Leroux independently has shown \ackermann-hardness of the VASS reachability
problem~\cite{JeromeAckermann}. He relies on similar known techniques, but his new contribution is substantially different than ours.

\paragraph*{Organisation of the paper}
In Section~\ref{sec:prelim} we introduce preliminary notions. Next, in Section~\ref{sec:outline} we present known approach
to the problem, introduce technical notions necessary to show our result and formulate the main technical Lemma~\ref{lem:amplifier}.
In Section~\ref{sec:zerotests} we present the main technique, which led to our result, namely the technique of performing
many zero-tests by using only one additional counter. We also present there two examples of application of this technique. The examples
are not necessary to understand the main construction, but are interesting in their own and introduce mildly the new technique.
In Section~\ref{sec:amplifier} we prove the main technical result, namely the Lemma~\ref{lem:amplifier}.
Finally, in Section~\ref{sec:future} we present possible future research directions.
\section{Preliminaries}\label{sec:prelim}

\paragraph*{Basic notions}
For $a, b \in \N$, $b \geq a$ we write $[a,b]$ for the set $\{a, a+1, \ldots, b-1, b\}$.
For a vector $v \in \Z^d$ and $i \in [1,d]$ we write $v[i]$ for the $i$-th entry of $v$.
For a vector $v \in \Z^d$ and the set of indices $S \subseteq [1,d]$ by $v[S] \in \Z^{|S|}$ we denote
vector $v$ restricted to the indices in $S$. By $0^d$ we represent the $d$-dimensional vector
with all entries being $0$.

\paragraph*{Vector Addition Systems}
A $d$-dimensional Vector Addition System with States ($d$-VASS)
consists of a finite set of \emph{states} $Q$ and a finite set of \emph{transitions} $T \subseteq Q \times \Z^d \times Q$.
Configuration of a VASS is a pair $(q, v) \in Q \times \N^d$, usually written $q(v)$. We write $\conf = Q \times \N^d$.
Transition $(p,t,q) \in T$ can be fired in the configuration $r(v) \in \conf$
if $p = r$ and $v+t \in \N^d$. Then we write $p(v) \trans{(p,t,q)} q(v+t)$.
The \emph{effect} of transition $(p, t, q)$ is a vector $t \in \N^d$, we write $\eff((p, t, q)) = t$.
A sequence of triples
$\rho = (c_1, t_1, c'_1), (c_2, t_2, c'_2), \ldots, (c_n, t_n, c'_n) \in \conf \times T \times \conf$
is a \emph{run} of VASS $V = (Q, T)$ if
for all $i \in [1,n]$ we have $c_i \trans{t_i} c'_i$ and for all $i \in [1,n-1]$ we have $c'_i = c_{i+1}$.
We extend naturally the definition of the effect to runs, $\eff(\rho) = t_1 + \ldots + t_n$.
Such a run $\rho$ is said to be \emph{from} the configuration $c_1$ \emph{to} the configuration $c'_n$.
We write then $c_1 \trans{\rho} c'_n$ slightly overloading the notation or simply $c_1 \reaches c'_n$
if there is some $\rho$ such that $c_1 \trans{\rho} c'_n$.
We also say then that the configuration $c_1$ \emph{reaches} the configuration $c'_n$
or $c'_n$ \emph{is reachable from} $c_1$.
By $\reach(\src, V) = \{c \mid \src \reaches c\}$ we denote the set of all the configurations reachable from configuration $\src$
and we call it the \emph{reachability set}.
We also write simply $\reach(\src)$ if VASS $V$ is clear from the context.
The following problem is the main focus of this paper.

\begin{quote}\label{qu:problem}
\textbf{Reachability problem for VASSes}
\begin{description}
  \item[Input] A VASS $V$ and two its configurations $\src, \trg$
  \item[Question] Does $\src \reaches \trg$ in $V$?
\end{description}
\end{quote}

The size of VASS $V$, denoted $\size(V)$, is the total number of bits needed to represent states and transitions of $V$.
A Vector Addition System (VAS) is a VASS with only one state (thus the state can be ignored).
It is folklore that reachability problems for VASSes and for VASes are interreducible in polynomial time,
therefore one can wlog. focus on one of them. In this paper we decide to work with VASSes as they
form a more robust model.

\paragraph*{Counter programs}
We often work with VASSes which have a special sequential form: each run of such a VASS
performs first some sequence of operations, then some other sequence of operations etc.
Such VASSes can be very conveniently described as counter programs.
A counter program is a sequence of instructions, each one being either the counter values modifications
of the form \add{\vr{x_1}}{a_1} \quad $\ldots$ \quad \add{\vr{x_d}}{a_d}
or a loop of the form
\begin{algorithmic}[1]
\Loop
\State P
\EndLoop
\end{algorithmic}
where $P$ is another counter program.
Such a counter program with $k$ instructions and $d$ counters $x_1, \ldots, x_d$
represents a $d$-VASS $V$ with states $q_1, \ldots, q_k, q_{k+1}$ (and some other ones)
such that:
\begin{itemize}
  \item there are two distinguished states of $V$, the state $q_1$ called the \emph{source} state of $V$
  and the state $q_{k+1}$ called the \emph{target} state of $V$;
  \item if the $i$-th instruction is of the form \add{\vr{x_1}}{a_1} \quad $\ldots$ \quad \add{\vr{x_d}}{a_d}
  then in $V$ there is a transition $q_i \trans{v} q_{i+1}$ where $v[j] = a_j$ if \add{\vr{x_j}}{a_j}
  is listed in the sequence of increments and $v[j] = 0$ otherwise;
  \item if the $i$-th instruction is the loop with body equal to counter program $P$ then
  in $V$ there are transitions $q_i \trans{0^d} \src_P$ and $\trg_P \trans{0^d} q_i$
  where $\src_P$ and $\trg_P$ are source and target states of VASS $V_P$ represented by program $P$
  \item if the $i$-th instruction is the loop then in $V$ there is a transition $q_i \trans{0^d} q_{i+1}$.
\end{itemize}
If the last instruction is a loop we often omit $q_{k+1}$ as the only transition incoming to it is
$q_k \trans{0^d} q_{k+1}$ and treat $q_k$ as the target state.

\begin{example}
The following counter program
\begin{algorithmic}[1]
\State \add{\vr{x}}{1}
\Loop
\State \sub{\vr{x}}{1} \quad \add{\vr{y}}{1}
\EndLoop
\Loop
\State \add{\vr{x}}{2} \quad \sub{\vr{y}}{1}
\EndLoop
\Loop
\State \sub{\vr{x}}{1} \quad \add{\vr{y}}{1}
\EndLoop
\Loop
\State \add{\vr{x}}{2} \quad \sub{\vr{y}}{1}
\EndLoop
\end{algorithmic}
represents the $2$-VASS presented below, state names are chosen arbitrary.
Notice that in the program there are five instructions: line 1 and loops in lines 2-3, 4-5, 6-7 and 8-9,
so the corresponding VASS has five states and can be depicted as follows.

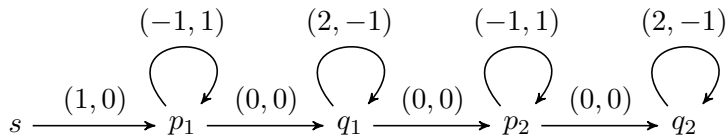
\begin{tikzpicture}[->,>=stealth',shorten >=1pt,auto,node distance=2.2cm,semithick]
\node (s) {$s$};
\node (p1) [right of=s] {$p_1$};
\node (q1) [right of=p1] {$q_1$};
\node (p2) [right of=q1] {$p_2$};
\node (q2) [right of=p2] {$q_2$};

\path[->]
(p1) edge[->, in=50, out=130, min distance=0.5cm,loop] node {$(-1,1)$}(p1)
(q1) edge[->, in=50, out=130, min distance=0.5cm,loop] node {$(2,-1)$}(q1)
(p2) edge[->, in=50, out=130, min distance=0.5cm,loop] node {$(-1,1)$}(p2)
(q2) edge[->, in=50, out=130, min distance=0.5cm,loop] node {$(2,-1)$}(q2)

(s) edge[->] node[above] {$(1,0)$} (p1)
(p1) edge[->] node[above] {$(0,0)$} (q1)
(q1) edge[->] node[above] {$(0,0)$} (p2)
(p2) edge[->] node[above] {$(0,0)$} (q2);
\end{tikzpicture}
\end{example}

\noindent
We often use macro \textbf{for} $i$ := $1$ \textbf{to} $n$ \textbf{do},
by which we represent just the counter program in which the body of the for-loop
is repeated $n$ times. We do not allow for the use of variable $i$ inside the for-loop.

\begin{example}\label{ex:2VASS}
The following counter program uses the macro \textbf{for}. For $n = 2$ it is equivalent to the above example.
\begin{algorithmic}[1]
\State \add{\vr{x}}{1}
\For {\, $i$ \, := \, $1$ \, \textbf{to } $n$}\label{l:for}
\Loop
\State \sub{\vr{x}}{1} \quad \add{\vr{y}}{1}
\EndLoop
\Loop
\State \add{\vr{x}}{2} \quad \sub{\vr{y}}{1}
\EndLoop
\EndFor
\end{algorithmic}

\noindent
The counter program represents the following $2$-VASS.

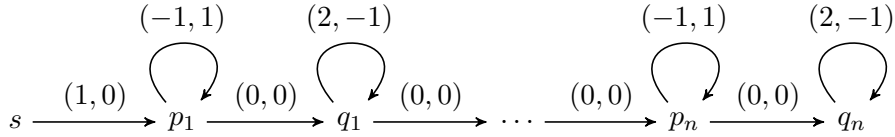
\begin{tikzpicture}[->,>=stealth',shorten >=1pt,auto,node distance=2.2cm,semithick]
\node (s) {$s$};
\node (p1) [right of=s] {$p_1$};
\node (q1) [right of=p1] {$q_1$};
\node (b) [right of=q1] {$\ldots$};
\node (pn) [right of=b] {$p_n$};
\node (qn) [right of=pn] {$q_n$};

\path[->]
(p1) edge[->, in=50, out=130, min distance=0.5cm,loop] node {$(-1,1)$}(p1)
(q1) edge[->, in=50, out=130, min distance=0.5cm,loop] node {$(2,-1)$}(q1)
(pn) edge[->, in=50, out=130, min distance=0.5cm,loop] node {$(-1,1)$}(pn)
(qn) edge[->, in=50, out=130, min distance=0.5cm,loop] node {$(2,-1)$}(qn)

(s) edge[->] node[above] {$(1,0)$} (p1)
(p1) edge[->] node[above] {$(0,0)$} (q1)
(q1) edge[->] node[above] {$(0,0)$} (b)
(b) edge[->] node[above] {$(0,0)$} (pn)
(pn) edge[->] node[above] {$(0,0)$} (qn);
\end{tikzpicture}
\end{example}

For a counter program $V$ we write $u_\inp \trans{V} u_\out$ if there is a run of $V$ starting in counter valuation
$u_\inp$ in the source state of $V$ and finishing in counter valuation $u_\out$ in the target state of $V$.

\paragraph*{Fast-growing functions and its complexity classes}
We introduce here a hierarchy of fast-growing functions and the corresponding complexity classes.
There are many known variants of the definition of the fast-growing function hierarchy.
Notice however, that the definition of the corresponding complexity classes $\F_i$
is robust and does not depend on the small changes in the definitions of the fast-growing hierarchy
(for the robustness argument see \cite[Section 4]{DBLP:journals/toct/Schmitz16}).

Let $F_1(n) = 2n$ and let $F_k(n) = \underbrace{F_{k-1} \circ \ldots \circ F_{k-1}}_{n}(1)$ for any $k > 1$.
Therefore we have $F_2(n) = 2^n$ and $F_3(n) = \underbrace{2^{2^{\iddots^2}}}_n = \tower(n)$.
We define the function $F_\omega$ as $F_\omega(n) = F_n(n)$.
The Ackermann function is defined exactly to be the $F_\omega$ function from the fast-growing hierarchy.

Based on functions $F_k$ we define complexity classes $\F_k$ also following definitions in~\cite{DBLP:journals/toct/Schmitz16}.
The complexity class $\F_k$ contains all the problems, which can be solved in time $f \circ g$,
where $g \in F_k$ and $f$ belongs to levels $F_i$ for $i < k$ closed under composition and limited primitive recursion.
The idea is that problems in $\F_k$ can be solved by some easier-then-$F_k$ reduction to an $F_k$-solvable problem.
For example the class $\F_3$, also called $\tower$, contains all the problems, which can be solved in the time $\tower(n)$,
but also for example those, which can be solved in time $\tower(2^{2^n})$,
as $\tower(2^{2^n}) = \tower(n) \circ 2^{2^n}$.
It is well known that complexity classes $\F_k$ have natural complete problems
concerning counter automata, which we formulate precisely in Section~\ref{sec:outline}.

\section{Outline}\label{sec:outline}
Here we outline the proof of our main result, Theorem~\ref{thm:main}.
We introduce gradually intuitions, which led us to this contribution.

\paragraph*{Counter automata and $\F_k$-hardness}
We follow some of the ideas of the previous lower bound result showing \tower-hardness~\cite{DBLP:conf/stoc/CzerwinskiLLLM19}.
In particular we reduce from a similar problem concerning counter automata.
Counter automata are extensions of VASSes in which transitions may have an additional condition 
that they are fired only if certain counter is equal to exactly zero.
Such transitions are called \emph{zero-tests}.
We say that a run of counter automaton is \emph{accepting} if it starts in the distinguished initial state with all counters equal to zero and finishes in the distinguished accepting state also with all counters equal to zero.
A run is $N$-bounded if all the counters along this run have values not exceeding $N$.
It is a folklore that the following problem is $\F_k$-hard for $k \geq 3$ (for a similar problem see~\cite[Section 2.3.2]{DBLP:journals/toct/Schmitz16}):

\begin{quote}\label{qu:problem}
\textbf{$F_k$-reachability for counter automaton}
\begin{description}
  \item[Input] Three-counter automaton $\A$, number $n \in \N$
  \item[Question] Does $\A$ have an $F_k(n)$-bounded accepting run?
\end{description}
\end{quote}

Our aim is to provide a polynomial time reduction, which for each $k \geq 3$, three-counter automaton $\A$ and number $n \in \N$
constructs a $6k$-VASS together with source and target configurations $\src$ and $\trg$ such that $\src \reaches \trg$
iff $\A$ has an $F_k(n)$-bounded accepting run. This will finish the proof of Theorem~\ref{thm:main}.

\paragraph*{Multiplication triples}
As suggested above the main challenge in showing $\F_k$-hardness is the need to simulate $F_k(n)$-bounded
counters and provide zero-tests for them. We first recall an idea from~\cite{DBLP:conf/stoc/CzerwinskiLLLM19}
which reduces the problem to constructing three counters with appropriate properties.
On an intuitive level the argument proceeds as follows. Assume a machine (in our case VASS) has access to triples
of the form $(M, y, My)$. Then it can use them to perform exactly $y$ sequences of actions, whatever these actions
exactly are, and in each sequence perform exactly $M$ actions. The idea is that in each sequence VASS decreases
the second counter by one therefore assuring that the number of sequences is exactly $y$. It uses the first counter
to assure that in each sequence the number of actions is at most $M$. During each action the third counter
is decreased by one, thus each sequence of actions decreases the third counter by at most $M$.
Therefore $y$ sequences of actions can decrease the third counter maximally by $My$ and moreover
if this counter was decreased by exactly $My$ it means that in every sequence exactly the maximal
possible number $M$ of actions was performed.
Thus by checking at the end of the whole process whether the second and the third counters are equal to zero we check
whether there were exactly $y$ sequences and in all the sequences there were exactly $M$ actions.
Below we exploit this idea roughly speaking for simulating zero-tests on counters which are bounded by some value $M$ 
an arbitrary number $y$ of times. So typically in our application $M$ will be bounded and $y$ is a guessed arbitrarily big number.
We also explain below more precisely what kind of VASS we need to prove
the $\F_k$-hardness of the reachability problem.

We say that a $(d+3)$-VASS $V$ for $d \geq 0$ together with its initial configuration $c$,
accepting state $q$ and a set of test counters $T \subseteq [1,d]$ is an \emph{$M$-generator} if:
\begin{itemize}
  \item all the configurations of the form $q(x, y, z, v)$ with $v \in \N^d$ in the set $\reach(c, V)$ such that $v[t] = 0$ for all $t \in T$ 
  fulfil $v = 0^d$, $x = M$ and $z = My$;
  \item for each $y \in \N$ we have $q(M, y, My, 0^d) \in \reach(c, V)$.
\end{itemize}
We call the counters $x$, $y$, $z$ the \emph{output counters}.
In other words an $M$-generator generates triples $(x, y, z)$ on its output counters
such that we are guaranteed that they are of the form $(M, y, My)$ and moreover each such triple can be generated.
We also say briefly that $(V, c, q, t)$ is an $M$-generator.

The following lemma shows that it is enough to focus on the construction of $M$-generators,
as they allow for simulation of $M$-bounded counters. The same idea and a similar statement was already present
in~\cite{DBLP:conf/stoc/CzerwinskiLLLM19}, but we prove the lemma in order to be self-contained.

\begin{lemma}\label{lem:generator}
For any $d$-VASS $(V, s, q, T)$ with $d \geq 12$ and $|T| \leq 4$, which is an $M$-generator,
and a three-counter automaton $\A$ one can construct in polynomial time
a $d$-VASS $V_\A$ with configurations $\src$ and $\trg$ such that $\src \reaches \trg$ iff $\A$ has an $M$-bounded accepting run.
\end{lemma}

\begin{proof}
The construction of the $d$-VASS $V_\A$ proceeds as follows.
The configuration $\src$ of $V_\A$ is the source configuration of the $M$-generator $(V, s, q, T)$.
We first run the $d$-VASS $(V, s, q, T)$, which outputs a triple $(c_1, c_2, c_3, 0^{d-3})$
under the condition that the test counters equal zero. For technical simplicity
we assume that the test counters are the ones with the biggest indices.
In the rest of the run we do not modify the test counters in order to assure (by setting $\trg[t] = 0$ for all $t \in T$)
that indeed these test counters equals zero at output of $V$. Recall that as $d \geq 12$ and we have at most $4$ test counters
then we are at least $8$ counters beside the test ones, namely: $c_1, \ldots, c_8$.
We need to simulate three counters of the automaton $\A$, say counters $x$, $y$ and $z$.
In order to assure that each run of $V_\A$ corresponds to an $M$-bounded run of $\A$ we add
for each counter $c$ another counter $\bar{c}$ such that at any time after an initialisation phase it holds $c + \bar{c} = M$.
We will use the counters $c_1$, $c_4$ and $c_5$ to simulate counters $x$, $y$ and $z$, respectively
and the counters $c_6$, $c_7$ and $c_8$ to simulate counters $\bar{x}$, $\bar{y}$
and $\bar{z}$, respectively. We thus need to set $c_6 = c_7 = c_8 = M$ in the initialisation phase,
which is realised by the following program fragment
\begin{algorithmic}[1]
\State \sub{\vr{c_2}}{1}
\Loop
  \State \sub{\vr{c_1}}{1} \quad \sub{\vr{c_3}}{1} \quad \add{\vr{c_6}}{1}  \quad \add{\vr{c_7}}{1} \quad \add{\vr{c_8}}{1}
\EndLoop
\end{algorithmic}
Counter $c_2$ is decreased here by $1$, while counter $c_3$ is decreased by at most value of $c_1$, which is $M$.
As explained before the only option for counter $c_3$ to reach value $0$ at the end of the run
is to match each decrease of $c_2$ by $1$ by a decrease by $M$. Therefore we are guaranteed
that in any run reaching configuration $\trg$ the initialisation phase indeed sets $\bar{x} = \bar{y} = \bar{z} = M$
and also we have $x = y = z = 0$ after this phase. The configuration $\trg$ of $V_\A$ is defined as follows:
the state corresponds to the accepting state of three-counter automaton $\A$, eight first counter values
correspond to counter values of accepting configuration of $\A$, namely $(c_1, \ldots, c_8) = (0^5, M, M, M)$
and all the other counters of $V_\A$ are equal zero in the configuration $\trg$.

Next VASS $V_\A$ simulates operations of counter automaton $\A$, namely increments, decrements and zero-tests.
Concretely speaking there is a copy of $\A$ inside of $V_\A$ with slightly modified transitions.
Simulation of operation $\coreadd{c}{a}$ (for both positive and negative $a$) in $\A$ is straightforward,
we add operations $\coreadd{c}{a}$ and $\coresub{\bar{c}}{a}$ to $V_\A$.
It is more challenging to simulate $\zerotest{c}$ in $\A$, we use the pair of counters $(c_2, c_3) = (b, Mb)$
generated by the $M$-generator for that.
Recall that using this pair we are able to perform exactly $b$ sequences of exactly $M$ actions. The idea is that for checking whether
$c = 0$ (and thus $\bar{c} = M$) we first transfer value of $\bar{c}$ to $c$ (i.e. decrement $\bar{c}$ and increment $c$) and
simultaneously decrement $c_3$. Then we transfer value of $c$ back to $\bar{c}$ also decrementing $c_3$.
In that way we can decrement $c_3$ at most $2M$ times and decrement by exactly $2M$ can happen only if the initial value of $c$
was zero and also final value of $c$ is zero as well. If we decrement $c_2$ by $2$ we assure that indeed $c_3$ needs to be decremented
by $2M$ and hence the $\zerotest{c}$ can be simulated as follows.
\begin{algorithmic}[1]
\State \sub{\vr{c_2}}{2}
\Loop
  \State \add{\vr{c}}{1} \quad \sub{\vr{\bar{c}}}{1} \quad \sub{\vr{c_3}}{1}
\EndLoop
\Loop
  \State \sub{\vr{c}}{1} \quad \add{\vr{\bar{c}}}{1} \quad \sub{\vr{c_3}}{1}
\EndLoop
\end{algorithmic}
Let us inspect the code to see that it indeed reflects the above story.
Recall that we keep all the time the invariant $c + \bar{c} = M$, so $\bar{c} \leq M$.
Therefore the loop in lines 2-3 is fired at most $M$ times.
Similarly the loop in lines 4-5 is fired at most $M$ times. Thus indeed the result of loops in lines 2-5 is the decrease
of counter $c_3$ by at most $2M$ and decrease by exactly $2M$ corresponds to initial and final value of $c$ being zero.
Thus lines 1-5 indeed simulate faithfully the zero-test.

As increments, decrements and zero-tests of $\A$ can be simulated faithfully by $V_\A$
one can see that runs from $\src$ to $\trg$ of $V_\A$
are in one-to-one correspondence with $M$-bounded runs of automaton $\A$.
\end{proof}

Our approach is therefore to construct $6k$-VASSes of not too big size which are $F_k(n)$-generators.

\paragraph*{Amplifiers}
At this moment it is natural to introduce a notion of \emph{amplifier}, which can be used
to produce an $N$-generator from an $M$-generator for $N$ much bigger than $M$.
For a function $f: \N \to \N$ we say that a $d$-VASS $V$ together with its \emph{input state} $p_\inp$,
\emph{output state} $p_\out$ and set of \emph{test counters} $T \subseteq [1,d]$ is an \emph{$f$-amplifier} if the following holds
\begin{itemize}
  \item if $p_\inp(M, x, Mx, 0^{d-3}) \reaches p_\out(v, b, y, z)$ for $v \in \N^d$ with $v[T] = 0$ then $v = 0^{d-3}$,
  $b = f(M)$ and $z = by$; and
  \item for each $y \in \N$ there exists an $x \in \N$ such that $p_\inp(M, x, Mx, 0^{d-3}) \reaches p_\out(0^{d-3}, f(M), y, f(M) \cdot y)$.
\end{itemize}
In other words, intuitively, if an amplifier inputs triples $(M, x, Mx)$ it outputs triples $(f(M), y, f(M) \cdot y)$ and moreover
each such triple can be output if an appropriate triple is delivered to the input.
In the above case we call the first three counters the \emph{input} counters and the last three counters the \emph{output} counters,
but in general we do not impose any order of the input and output counters.
Notice that notions of amplifier and generators are very much connected as suggested by the following claim.

\begin{lemma}\label{lem:amplifier-to-generator}
For any $d \geq 3$ if there exists a $d$-dimensional $f$-amplifier $V$
then exists a $d$-dimensional $f(M)$-generator of size linear in $\size(V)+\log(M)$ with the same number of test counters as $V$.
\end{lemma}

\begin{proof}
We construct the $f(M)$-generator as follows. In its initial state we have a loop with the effect $(0, 1, M, 0^{d-3})$,
thus after $x$ applications of it we get a vector $(0, x, Mx, 0^{d-3})$. Then a transition with effect $(M, 0, 0, 0^{d-3})$ leads to the input state
of the $f$-amplifier, so the $f$-amplifier inputs a tuple $(M, x, Mx, 0^{d-3})$.
Immediately from the definition of an amplifier we get that all the runs reaching the output state of the amplifier
with vectors of the form $(v, x, y, z)$ with $v \in \N^{d-3}$ such that $v[t] = 0$ for all test counters $t \in T$
fulfil $v = 0^{d-3}$, $x = f(M)$, $z = f(M) \cdot y$ and additionally such configurations for all $y \in \N$ can be reached,
which finishes the proof.
\end{proof}

Observe that taking into account Lemmas~\ref{lem:generator} and \ref{lem:amplifier-to-generator} in order to prove
Theorem~\ref{thm:main} it is enough to show the following lemma.

\begin{lemma}\label{lem:amplifier}
For each $k \geq 1$ there exists a $6k$-VASS of size exponential in $k$ (at most $C^k$ for some constant $C \in \N$)
which is an $F_k$-amplifier with at most four test counters.
\end{lemma}

The advantage of amplifiers over generators is that we can easily compose them.
Notice that having two VASSes: a $(d_1+3)$-VASS being an $f_1$-amplifier
and a $(d_2+3)$-VASS being an $f_2$-amplifier it is easy
to construct a $(d_1+d_2+3)$-VASS being an $f_1 \circ f_2$-amplifier just by using sequential composition of the $f_2$-amplifier
and $f_1$-amplifier. However, the drawback of the construction is that the dimension grows substantially. The main challenge
in the proof of Lemma~\ref{lem:amplifier} is to build amplifiers for much bigger functions from amplifiers for much smaller functions
without adding too many new counters. The proof of Lemma~\ref{lem:amplifier} is presented in Section~\ref{sec:amplifier}.

\paragraph*{Big counter values}
In order to prove $F_k$-hardness for VASS reachability problem one should in particular
construct VASSes $V_k$ in which the shortest run from some source configuration to some target configuration
has length at least $F_k(n)$, where $n = \size(V_k)$. Notice that some configuration on such a run needs to have
some counters of value at least roughly $F_k(n)$. It has been known since a long time that such VASSes exist
and it is relatively easy to construct them. The hard part is to design a VASS, in which \emph{every} run reaches
high counter values and on a very high level of abstraction one can see our construction as mainly achieving this goal.

Below we present an example family of VASSes $V_k$ in growing dimension $k+1$ with reachability sets being finite,
but which can reach values of counters up to $F_k(n)$.
Knowing this construction is absolutely not needed to understand our construction and this part of the paper
can be omitted during the reading without any harm. We present it however as we believe
that it helps to distinguish which parts of the proof of Lemma~\ref{lem:amplifier} are pretty standard and which were the real challenge.
In short words the real challenge was to force the runs to have values zero at some precise points,
we show in Section~\ref{sec:zerotests} in details how to guarantee this.

For $d = 2$ we have the following $3$-VASS $V_2$.
\begin{algorithmic}[1]
\Loop
  \Loop
    \State \add{\vr{x_1}}{2} \quad \sub{\vr{x_2}}{1}
  \EndLoop
  \Loop
    \State \sub{\vr{x_1}}{1} \quad \add{\vr{x_2}}{1}
  \EndLoop
  \State \sub{\vr{x_{3}}}{1}
\EndLoop
\end{algorithmic}

In general we construct $(k+1)$-dimensional VASS $V_k$ in the following way from $k$-dimensional VASS $V_{k-1}$.

\begin{algorithmic}[1]
\Loop
  \State $V_{k-1}$
  \Loop
    \State \sub{\vr{x_1}}{1} \quad \add{\vr{x_k}}{1}
  \EndLoop
  \State \sub{\vr{x_{k+1}}}{1} \quad \add{\vr{x_3}}{1} \quad \ldots \quad \add{\vr{x_{k-1}}}{1}
\EndLoop
\end{algorithmic}

It is quite easy to see that the reachability set of $V_k$ is finite when starting from any counter
valuation and one can show it easily by induction on $k$. Therefore it remains to show the following proposition.

\begin{proposition}
For each $k \in \N$ it holds
$(1, 0, 1^{k-2}, m-1) \trans{V_k} (F_k(m), 0^k)$.
\end{proposition}

\begin{proof}
We show the proposition by induction on $k$.
We start the induction from $k  = 2$, one can easily see  that indeed $(1, 0, m-1) \trans{V_2} (2^m, 0, 0)$.
For the induction step assume that $(1,0,1^{k-3},m-1) \trans{V_{k-1}} (F_{k-1}(m), 0^{k-1})$.
We therefore have for any $\ell \in \N$ that
\begin{align*}
(1, 0, 1^{k-3}, m-1, \ell) & \trans{(2)} (F_{k-1}(m), 0^{k-1}, \ell) \trans{(3-4)} (1, 0^{k-2}, F_{k-1}(m)-1, \ell) \\
& \trans{(5)} (1, 0, 1^{k-3}, F_{k-1}(m)-1, \ell-1),
\end{align*}
where by $\trans{(i)}$ we denote the transformation on counters caused by line $i$ of program $V_k$.
Therefore we have
\begin{align*}
(1, 0, 1^{k-3}, 1, m-1) & \trans{(2-5)} (1, 0, 1^{k-3}, F_{k-1}(1)-1, m-1) \\
& \trans{(2-5)} (1, 0, 1^{k-3}, F_{k-1} \circ F_{k-1}(1)-1, m-2) \\
& \trans{(2-5)} \ldots \trans{(2-5)} (1, 0, 1^{k-3}, \underbrace{F_{k-1} \circ \ldots \circ F_{k-1}}_{m-1}(1)-1, 0) \\
& \trans{(2)} (\underbrace{F_{k-1} \circ \ldots \circ F_{k-1}}_m(1), 0^k) = (F_k(m), 0^k).
\end{align*}
which finishes the induction step.
\end{proof}

\section{Zero-tests}\label{sec:zerotests}
The main contribution of this paper is a novel technique of zero-testing, which allows for performing many zero-tests
simultaneously. This will be the key technique in the proof of Lemma~\ref{lem:amplifier}.
We prefer to introduce it mildly before using it in Section~\ref{sec:amplifier} and present first how it works on a few simple examples.
Already on these examples its power is visible.

\begin{example}
Imagine first that we are given a VASS run $\rho$ and we want to test some counter $x$ for being exactly zero in three moments
along this run: in configurations $c_1$, $c_2$ and $c_3$. Assume that value of $x$ is zero at the beginning of $\rho$.
Let value of $x$ in these configurations be $x_1$, $x_2$ and $x_3$, respectively.
A naive way to solve this problem is to add three new counters, which are copies of $x$, but the first one stops
copying effects of transitions on $x$ in $c_1$, the second one in $c_2$ and the third one in $c_3$.
In that way the additional counters keep values of $x_1$, $x_2$ and $x_3$ till the end of the run and can be checked
there for zero (just by setting the target configuration to zero on these counters).
We show here how to perform these three zero-tests using just one additional counter, we call it the \emph{controlling counter}
and say that it \emph{controls} the other counters.
Let $\rho_1$ be the part of $\rho$ before $c_1$, $\rho_2$ the part in between $c_1$ and $c_2$ and $\rho_3$
the part in between $c_2$ and $c_3$.
Let $y_1$, $y_2$ and $y_3$ be the effects of $\rho_1$, $\rho_2$ and $\rho_3$ on $x$, respectively.
We can easily see that $x_1 = y_1$, $x_2 = y_1 + y_2$ and $x_3 = y_1 + y_2 + y_3$.
Notice that we can check whether all the $x_1$, $x_2$ and $x_3$ are equal to zero
by checking whether its sum $x_1 + x_2 + x_3$ equals zero,
as all the $x_i$ are nonnegative. We have $x_1 + x_2 + x_3 = 3 y_1 + 2 y_2 + y_3$, so it is enough to check
whether $3 y_1 + 2 y_2 + y_3 = 0$. Instead of adding three new counters we add only one,
which computes the value $3 y_1 + 2 y_2 + y_3$.
We realise it in the following way.
Every increase of $x$ by $a$ during $\rho_1$ is reflected by an increase of the controlling counter by $3a$.
Similarly, an increase of $x$ by $a$ on $\rho_2$ is reflected by an increase of the controlling counter by $2a$ and on $\rho_3$ just by $a$.
After configuration $c_3$ the controlling counter is not modified.
Therefore testing the controlling counter for $0$ at the end of the run (by setting appropriately the target configuration on that counter)
checks indeed whether $x_1 = x_2 = x_3 = 0$.
\end{example}

This approach can be generalised to more zero-tests and moreover to zero-tests on different counters,
as shown by the following lemma. In the lemma the aim of the last (controlling) counter is to allow
for zero-testing the $i$-th counter on configurations with indices in $S_i$, for all $i \in [1,d]$.
Notice also that configurations $c_i$ are not necessarily all the configurations on the run $\rho$, but some subset,
which can be restricted only to those in which we want to perform some zero-test.

\begin{lemma}\label{lem:zero-testing}
Let $\src \trans{\rho} \trg$ be a run of a $(d+1)$-VASS $V$
and let $\src = c_0, c_1, \ldots, c_{n-1}, c_n = \trg$ be some of the configurations on $\rho$.
Let $\rho_j$ for $j \in [1,n]$ be the parts of the run $\rho$ starting in $c_{j-1}$ and finishing in $c_j$, namely
\[
c_0 \trans{\rho_1} c_1 \trans{\rho_2} \ldots \trans{\rho_{n-1}} c_{n-1} \trans{\rho_n} c_n.
\]
Let $S_1, \ldots, S_d \subseteq [0,n]$ be the sets of indices of $c_j$, in which we want to zero-test counters numbered $1, \ldots, d$, respectively
and let $N_{j,i} = |\{k \geq j \mid k \in S_i\}|$ for $i \in [1,d], j \in [0,n]$ be the number of zero-tests, which we want to perform
on the $i$-th counter starting from configuration $c_j$ (in other words after the run $\rho_j$ for $j > 0$).
Then if:
\begin{enumerate}
  \item[(1)] $\src[d+1] = \sum_{i=1}^d N_{0,i} \cdot \src[i]$;
  \item[(2)] for each $j \in [1,n]$ we have $\eff(\rho_j, d+1) = \sum_{i=1}^d N_{j,i} \cdot \eff(\rho_j, i)$; and
  \item[(3)] $\trg[d+1] = 0$
\end{enumerate}
then for each $i \in [1,d]$ and for each $j \in S_i$ we have $c_j[i] = 0$.
\end{lemma}

Before we proceed with the proof of Lemma~\ref{lem:zero-testing} we comment a bit on the lemma
and see how the example above fulfils its conditions. The condition (1) assures that at the configuration $\src$
the controlling $(d+1)$-th counter is appropriately related to the other ones. In the example it is trivially fulfilled by the equation $0 = 0$.
The condition (2) assures that the effect of each part of the run $\rho$ is appropriately reflected by the controlling counter.
On the example it is fulfilled as each change of the $x$-counter by $a$ on $\rho_1$, $\rho_2$ and $\rho_3$
is reflected by the change of controlling counter by $3a$, $2a$ and $a$, respectively.
Notice that in the case of our example the condition (2) is satisfied even if we set $\rho_j$ to be single transitions
and $n$ equals length of the run.
However in many applications of Lemma~\ref{lem:zero-testing} condition (2) is true only if we set $\rho_j$ to be some groups of transitions
forming a longer sub-run of $\rho$, this is why we formulate the lemma in such a stronger version.
Condition (3) assures that the controlling counter is indeed zero at the end of the run.

\begin{proof}
Notice first that in order to check whether for each $i \in [1,d]$ and for each $j \in S_i$ we have $c_j[i] = 0$ it is enough
to check whether the sum of all $c_j[i]$ is zero, namely to check whether
$\textsc{sum} = \sum_{i \in [1,d], j \in S_i} c_j[i] = 0$.

For each $i \in [1,d]$ and $j \in S_i$ the value $c_j[i]$ is the sum of the initial value of the $i$-th counter
and effects of all the runs $\rho_k$ before configuration $c_j$ on the $i$-th counter.
In other words $c_j[i] = c_0[i] + \sum_{k=1}^j \eff(\rho_k, i)$.
We therefore have
\[
\textsc{sum} = \sum_{i \in [1,d], j \in S_i} c_j[i] =  \sum_{i \in [1,d], j \in S_i} (c_0[i] + \sum_{k=1}^j \eff(\rho_k, i)).
\]
In the rightmost expression $c_0[i]$ occurs exactly $N_{0,i}$ times and each $\eff(\rho_k, i)$ occurs exactly $N_{k,i}$ times.
Therefore
\[
\textsc{sum} = \src[d+1] + \sum_{k=1}^n \eff(\rho_k, d+1) = \trg[d+1],
\]
where the first equation follows from (1) and (2).
By (3) we have $\trg[d+1] = 0$ which implies that indeed $\textsc{sum} = 0$ and finishes the proof.
\end{proof}

Below we present two examples of the application of Lemma~\ref{lem:zero-testing}:
an example of a $3$-VASS with transitions represented in unary (shortly unary $3$-VASS) with exponential shortest run
and an example of a $7$-VASS with transitions represented in binary (shortly binary $7$-VASS) with doubly-exponential shortest run.
Low dimensional unary VASSes with exponential shortest runs and binary VASSes
with doubly-exponential shortest run have been presented in~\cite{DBLP:conf/concur/Czerwinski0LLM20}.
We present here the new examples in order to illustrate our technique, but also to provide
another set of nontrivial VASS examples in low dimensions.

\begin{example}\label{ex:3VASS}
We consider the $2$-VASS from Example~\ref{ex:2VASS} with one controlling counter added, the counter $c$.
We analyse runs starting from all zeros and finishing in all zeros.
We will observe that adding the controlling counter forces the loops in lines 3-4 (or states $p_i$) and in lines 5-6 (or states $q_i$)
to be performed maximal possible number of times, namely the $x$ counter is zero when run leaves line 4 and the $y$ counter is zero
when run leaves line 6. We have added comments in the program below to emphasise where which counters are forced to be zero.
This forces the run to visit configuration $(2^n, 0, 0)$ in line 6 and therefore to be exponential in VASS size.
\renewcommand{\algorithmiccomment}[1]{\bgroup\hfill//~#1\egroup}
\begin{algorithmic}[1]
\State \add{\vr{x}}{1} \quad \add{\vr{c}}{n}
\For {\, $i$ \, := \, $1$ \, \textbf{to } $n$}\label{l:for}
\Loop \label{l:weakly1}
\State \sub{\vr{x}}{1} \quad \add{\vr{y}}{1} \algorithmiccomment{zero-test on $x$ after the loop}
\EndLoop
\Loop
\State \add{\vr{x}}{2} \quad \sub{\vr{y}}{1} \quad \add{\vr{c}}{n-i-1} \algorithmiccomment{zero-test on $y$ after the loop}
\EndLoop \label{l:weakly2}
\EndFor
\Loop
\State \sub{\vr{x}}{1}
\EndLoop
\end{algorithmic}

\noindent
To see the comparison with Example~\ref{ex:2VASS} we also present our VASS in a more traditional way.
For simplicity we do not write labels if they are vectors with all entries being zero.

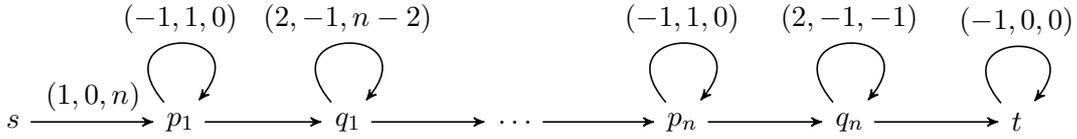
\begin{tikzpicture}[->,>=stealth',shorten >=1pt,auto,node distance=2.2cm,semithick]
\node (s) {$s$};
\node (p1) [right of=s] {$p_1$};
\node (q1) [right of=p1] {$q_1$};
\node (b) [right of=q1] {$\ldots$};
\node (pn) [right of=b] {$p_n$};
\node (qn) [right of=pn] {$q_n$};
\node (t) [right of=qn] {$t$};

\path[->]
(p1) edge[->, in=50, out=130, min distance=0.5cm,loop] node {$(-1,1,0)$}(p1)
(q1) edge[->, in=50, out=130, min distance=0.5cm,loop] node {$(2,-1,n-2)$}(q1)
(pn) edge[->, in=50, out=130, min distance=0.5cm,loop] node {$(-1,1,0)$}(pn)
(qn) edge[->, in=50, out=130, min distance=0.5cm,loop] node {$(2,-1,-1)$}(qn)
(t) edge[->, in=50, out=130, min distance=0.5cm,loop] node {$(-1,0,0)$}(t)

(s) edge[->] node[above] {$(1,0,n)$} (p1)
(p1) edge[->] (q1)
(q1) edge[->] (b)
(b) edge[->] (pn)
(pn) edge[->] (qn)
(qn) edge[->] (t);
\end{tikzpicture}

\noindent
Any run from $s(0,0,0)$ to $t(0,0,0)$ crosses through states $p_1, \ldots, q_n$ and therefore can be divided
into $2n+1$ parts $\rho_1, \ldots, \rho_{2n+1}$ by cutting in the last configurations in appropriate states.
We want the counter $x$ to be zero-tested in states $p_1, \ldots, p_n$
and counter $y$ to be zero-tested in states $q_1, \ldots, q_n$, so we set
$S_1 = \{1, 3, \ldots, 2n-1\}$
and $S_2 = \{2, 4, \ldots, 2n\}$.
One can easily check that the considered VASS fulfils conditions of Lemma~\ref{lem:zero-testing}.
Condition (1) is trivially fulfilled as $0 = 0$ and condition (3) is fulfilled as we demand to reach $t(0,0,0)$,
so the controlling counter is equal to zero at the target configuration.
In order to see that condition (2) is also fulfilled we show that it is even fulfilled for each transition (not only
for runs in between of zero-tests). Let us consider four lines: 1, 4, 6 and 8.
In line 1 counter $x$ is awaiting $n$ zero-tests, so increase of $x$ by $1$ should
be reflected by increase of $c$ by $n$. In line 4 counter $x$ is awaiting $(n+1)-i$ zero-tests and so similarly counter $y$.
Therefore in line 4 counter $c$ should be increased by $(n+1-i) \cdot 1 + (n+1-i) \cdot (-1) = 0$.
In line 6 counter $x$ is awaiting $n-i$ zero-tests, while counter $y$ is awaiting $(n+1)-i$ zero-tests, therefore counter $c$
should be increased here by $(n-i) \cdot 2 + (n+1-i) \cdot (-1) = 2n - 2i - n - 1 + i = n-1-i$. In line 8 counter $x$ is not awaiting for
any zero-test, so counter $c$ should not be changed because of its changes.
Summarising all of that we see that $c = 0$ at the target configuration forces the run to perform maximal number of times
loops in the states $p_i$ and $q_i$ and therefore there is only one run of our $3$-VASS, which in particular visits the configuration
$q_n(2^n, 0, 0)$ and has exponential length.
\end{example}

\begin{example}\label{ex:7VASS}
Here we present an example of a binary $7$-VASS with doubly-exponential shortest run.
This result is not needed in the proof of Lemma~\ref{lem:amplifier}, we show it however in order to illustrate how
to use the technique of performing many zero-tests in the case when the number of zero-tests is bigger then the size of the VASS.
In Example~\ref{ex:3VASS} the number of zero-tests both on $x$ and $y$ counters was comparable to the size of the VASS.
Therefore different behaviour of controlling $c$ in different phases of the run could be implemented
by different behaviour of $c$ in different states. In the current example this is not possible.
Let us recall the well known Hopcroft-Pansiot example of a $3$-VASS from~\cite{DBLP:journals/tcs/HopcroftP79}.
As $n$ is given in binary it can have a doubly-exponential run.
\begin{algorithmic}[1]
\State \add{\vr{x}}{1} \quad \add{\vr{z}}{n}
\Loop
\Loop
\State \sub{\vr{x}}{1} \quad \add{\vr{y}}{1}
\EndLoop
\Loop
\State \add{\vr{x}}{2} \quad \sub{\vr{y}}{1}
\EndLoop
\State \sub{\vr{z}}{1}
\EndLoop
\end{algorithmic}
We can observe that there is a run which finishes with counter values $(x, y, z) = (2^n, 0, 0)$ and $2^n$ is doubly-exponential
wrt. the VASS size. Notice however that nothing forces the run to reach so high value of $x$.
Our aim is now to add a controlling counter $c$ which would force the loops in lines 3-4 and in lines 5-6 to be applied maximal number of times. In the case when $z = 0$ at the end of the run we know that the main loop in lines 2-7 is executed exactly $n$ times,
therefore we want to test both $x$ and $y$ exactly $n$ times for zero. It is easy to observe that in lines 3-4 both
counters are awaiting $z$ zero-tests and in lines 5-6 counter $x$ is awaiting $(z-1)$ zero-tests, while $y$ is awaiting $z$ zero-tests.
Therefore the correct updates on controlling counter $c$ should be: increase by $0$ in line $4$, and increase by $(z-1) \cdot 2 + z \cdot (-1) = z-2$ in line $6$. The counter program thus should be the following.
\begin{algorithmic}[1]
\State \add{\vr{x}}{1} \quad \add{\vr{z}}{n} \quad \add{\vr{c}}{n}
\Loop
\Loop
\State \sub{\vr{x}}{1} \quad \add{\vr{y}}{1}
\EndLoop
\Loop
\State \add{\vr{x}}{2} \quad \sub{\vr{y}}{1} \quad \add{\vr{c}}{\vr{z}-2}
\EndLoop
\State \sub{\vr{z}}{1}
\EndLoop
\end{algorithmic}
One can however easily observe that the operation \add{\vr{c}}{\vr{z}-2} is not a valid VASS operation, as $\vr{z}-2$ is not a constant.
Fortunately counter $\vr{z}$ is a counter bounded by $n$ and in that case we can implement this operation using only VASS operations.
Very intuitively in order to implement \add{\vr{c}}{\vr{z}} it is enough to decrement $\vr{z}$ from its current value to value zero
and simultaneously increment value of counter $\vr{c}$. In order to be able to restore the original value of $\vr{z}$ and to be able
to check whether $\vr{z}$ reached zero we need to add auxiliary counters.
We add three additional counters $\vr{z'}$, $\vr{\bar{z}}$ and $\vr{\bar{z}'}$ such that all the time after line $1$ we keep
invariants $\vr{z} + \vr{\bar{z}} = n$ and $\vr{z'} + \vr{\bar{z}'} = n$. Notice that with those invariants we can easily check whether
$\vr{z} = 0$ just by performing two operations: \sub{\vr{\bar{z}}}{n} and then \add{\vr{\bar{z}}}{n}, similarly we test whether $\vr{z'} = 0$.
We introduce a macro for these operations here, writing \zerotest{\vr{z}} and \zerotest{\vr{z'}}.
These zero-tests should not be confused with zero-tests on counters $\vr{x}$ and $\vr{y}$. Notice that they are of a very different
nature and we implement \zerotest{\vr{z}} without using controlling counter $\vr{c}$, but by the use of the fact that $\vr{z}$ is $n$-bounded.

Therefore in line 1 we have now additionally operation \add{\vr{\bar{z}'}}{n},
in line 7 additionally operation \add{\vr{\bar{z}}}{1}
and in line 6 instead of operation \add{\vr{c}}{\vr{z}-2} we place the following code of VASS
\begin{algorithmic}[1]
\Loop
\State \add{\vr{c}}{1} \quad \sub{\vr{z}}{1} \quad \add{\vr{z'}}{1}
\State \add{\vr{\bar{z}}}{1} \quad \sub{\vr{\bar{z}'}}{1}
\EndLoop
\State \zerotest{\vr{z}}
\Loop
\State \add{\vr{z}}{1} \quad \sub{\vr{z'}}{1}
\State \sub{\vr{\bar{z}}}{1} \quad \add{\vr{\bar{z}'}}{1}
\EndLoop
\State \zerotest{\vr{z'}}
\State \sub{\vr{c}}{2}
\end{algorithmic}
The aim of line 2 is to perform \add{\vr{c}}{\vr{z}} and keep $\vr{z} + \vr{z'}$ constant. We need to keep $\vr{z} + \vr{z'}$ constant
to be able to reconstruct later the original value of $\vr{z}$.
In line 3 we update $\vr{\bar{z}}$ and $\vr{\bar{z}'}$ to keep the invariants and in line 4 we check whether indeed we have added
everything from $\vr{z}$ to $\vr{c}$. Lines 5-8 are devoted to moving values of $\vr{z}$ and $\vr{z'}$ back to original ones,
while the line 9 takes care of subtracting $2$ from $\vr{c}$, as our aim is to perform \add{\vr{c}}{\vr{z}-2}, not \add{\vr{c}}{\vr{z}}.

One can easily check that Lemma~\ref{lem:zero-testing} applies to our situation when one defines parts of the run $\rho_j$
corresponding to the single lines of the considered $7$-VASS.
Therefore if we demand $c = 0$ and $z = 0$ after the line 7 of the $7$-VASS then both the loops in lines 3-4 and in lines 5-6 have to be executed
maximally each time.
Notice that here we need to consider run fragments $\rho_j$ in Lemma~\ref{lem:zero-testing} to be longer than single transitions.
For example all the operations \add{\vr{x}}{2}, \sub{\vr{y}}{1} and \add{\vr{c}}{\vr{z}-2}
should be contained in one fragment $\rho_j$ and operation \add{\vr{c}}{\vr{z}-2} is implemented as long sequence of transitions, as presented above.
Summarising, by Lemma~\ref{lem:zero-testing} at the end of the run we have $x = 2^n$, which is indeed doubly-exponential in VASS size
since $n$ is encoded in binary.
\end{example}

\section{Amplifiers}\label{sec:amplifier}
This section is devoted to the proof of Lemma~\ref{lem:amplifier}.
Recall that we need to show that for each $k \geq 1$ there exists a $6k$-VASS of size exponential in $k$
which is an $F_k$-amplifier. We prove it by induction on $k$ with induction assumption additionally strengthened
by the fact that the test counters include all the input counters and at most one additional counter.
For $k = 1$ it is not hard to construct a $6$-VASS, which is an $F_1$-amplifier, recall that $F_1(n) = 2n$.
The following VASS realises our goal, the input counters are $x_1$, $x_2$ and $x_3$,
the output counters are $x_4$, $x_5$ and $x_6$ and the test counters are only the input counters.

\begin{algorithmic}[1]
\Loop
  \State \sub{\vr{x_2}}{2} \quad \add{\vr{x_5}}{1}
  \Loop
    \State \sub{\vr{x_1}}{1} \quad \add{\vr{x_4}}{1} \quad \sub{\vr{x_3}}{1} \quad \add{\vr{x_6}}{1}
  \EndLoop
  \Loop
    \State \add{\vr{x_1}}{1} \quad \sub{\vr{x_4}}{1} \quad \sub{\vr{x_3}}{1} \quad \add{\vr{x_6}}{1}
  \EndLoop
\EndLoop
\State \sub{\vr{x_2}}{1}
\Loop
  \State \sub{\vr{x_1}}{1} \quad \add{\vr{x_4}}{2} \quad \sub{\vr{x_3}}{1}
\EndLoop
\end{algorithmic}

Lines 1-6 are devoted to set the correct values of $x_5$ and $x_6$,
while lines 7-9 set the correct value of $x_4$.
Assume that at the input we have $(x_1, x_2, x_3) = (n, x, nx)$ and recall that initially $x_4 = x_5 = x_6 = 0$.
The proof idea is similar as in the proof of Lemma~\ref{lem:generator}, triple $(n, x, nx)$ is used
to perform exactly $x$ sequences of exactly $n$ actions.
Observe first that until line 8 the sum $x_1 + x_4$ does not change, thus we have $x_1 + x_4 = n$.
This means that loops in lines 3-4, 5-6 and 8-9 all can be fired at most $n$ times. Each such a loop
corresponds to one operation \sub{\vr{x_2}}{1} (in lines 2 or 7) and at most $n$ operations
\sub{\vr{x_3}}{1} (in lines 4, 6 and 9). This means that in order to reach $x_3 = 0$ at the end of the run
each loop has to be fired exactly $n$ times. Moreover the loop in lines 1-6 has to be fired exactly $\frac{x-1}{2}$
times, as the final value of $x_2$ also needs to be $0$.
Therefore final values of $(x_4, x_5, x_6)$ are $(2n, \frac{x-1}{2}, 2n \cdot \frac{x-1}{2})$, which finishes the proof for $k=1$.

For an induction step assume that $V_{k-1}$ is a $(6k-6)$-dimensional VASS of size exponential in $k-1$ and an $F_{k-1}$-amplifier.
We aim at constructing a $6k$-VASS, which is an $F_k$-amplifier and its size is at most $C \cdot \size(V_{k-1})$
for some constant $C$ which does not depend on $k$.
The idea to obtain $F_k$-amplifier is the following: start from the triple $(1, x, x)$ and apply the $F_{k-1}$-amplifier
$n$ times in a row, where $n$ is the input value. The main challenge is to achieve it without adding new counters
for each application. We show here how we obtain it by adding only six new counters. We crucially rely on
the Lemma~\ref{lem:zero-testing}.

The $F_k$-amplifier has the following $6k$ counters: input triple $(i_1, i_2, i_3)$,
output triple $(o_1, o_2, o_3)$, an auxiliary triple $(s_1, s_2, s_3)$, controlling counter $c$, two auxiliary
counters $y_1$, $y_2$ and $6k-12$ counters, which are the counters of $V_{k-1}$ being neither its input nor output counters.
The triple $(s_1, s_2, s_3)$ will be used inside the $F_k$-amplifier as an input triple of a $F_{k-1}$-amplifier.
The test counters are the input counters $(i_1, i_2, i_3)$ and the controlling counter $c$.

We first present the code for an $F_k$-amplifier $V_k$, which uses illegal constructions
like ''\textbf{for }$i$:= $1$ \textbf{to} $n$ \textbf{do}''
or ''\add{\vr{c}}{z}'' where $z$ is current value of another counter, in order to provide an intuition what $V_k$ does.
Then we show how we can implement the mentioned constructions using only legal VASS operations.
Assume that input counter values on $(i_1, i_2, i_3)$ are $(n, x, nx)$. Thus
we aim at producing on output counters $(o_1, o_2, o_3)$ values $(F_k(n), m, F_k(n) \cdot m)$ for some $m \in \N$.
Let $V^{[i]}_{k-1}$ be the modified version of $V_{k-1}$ in which the controlling counter $c$
is also appropriately modified: each modification \add{\vr{x}}{1} (or \sub{\vr{x}}{1})
for counter $\vr{x}$ being any output or test counter of $V_{k-1}$
is accompanied with a modification \add{\vr{c}}{(n+1)-i} (or \sub{\vr{c}}{(n+1)-i}).
It is important to emphasise that all the counters of $V^{[i]}_{k-1}$ as well as the counter $\vr{c}$
are shared among $V^{[i]}_{k-1}$ for different $i$. Recall that the test counters of $V^{[i]}_{k-1}$ contain
all the input counters $s_1$, $s_2$, $s_3$ and possibly one additional counter.
\begin{algorithmic}[1]
\State \add{\vr{s_1}}{1}  \quad \add{\vr{c}}{n}
\Loop
\State \add{\vr{s_2}}{1} \quad \add{\vr{s_3}}{1} \quad \add{\vr{c}}{2n}
\EndLoop
\For {\, $i$ \, := \, $1$ \, \textbf{to } $n$}\label{l:for}
  \State $V^{[i]}_{k-1}(s_1, s_2, s_3, o_1, o_2, o_3)$
  \Loop
    \State \sub{\vr{o_1}}{1} \quad \add{\vr{s_1}}{1} \quad \sub{\vr{c}}{1}
  \EndLoop
  \Loop
    \State \sub{\vr{o_2}}{1} \quad \add{\vr{s_2}}{1} \quad \sub{\vr{c}}{1}
  \EndLoop
  \Loop
    \State \sub{\vr{o_3}}{1} \quad \add{\vr{s_3}}{1} \quad \sub{\vr{c}}{1}
  \EndLoop
\EndFor
\end{algorithmic}

\begin{claim}
The above counter program is an $F_k$-amplifier.
\end{claim}

\begin{proof}
The aim of lines 1-3 is to set the triple $(s_1, s_2, s_3)$ to values $(1, a_0, a_0)$ for some arbitrary guessed $a_0 \in \N$.
For a function $f: \N \to \N$ let us denote by $f^{(m)}(n)$ the $m$-fold application of $f$ to $n$.
Then in lines 4-11 we perform $n$ times, for $i \in \{1, \ldots, n\}$ the following operations:
\begin{itemize}
  \item in line 5 from a triple $(F_{k-1}^{(i-1)}(1), a_{i-1}, F_{k-1}^{(i-1)}(1) \cdot a_{i-1})$
  for some $a_{i-1} \in \N$ on counters $(s_1, s_2, s_3)$
  the $F_{k-1}$-amplifier $V^{[i]}_{k-1}$ computes a triple $(F_{k-1}^{(i)}(1), a_i, F_{k-1}^{(i)}(1) \cdot a_i)$ for some $a_i \in \N$
  on counters $(o_1, o_2, o_3)$ under the condition that the test counters of $V^{[i]}_{k-1}$ are zero after its run;
  \item in lines 6-11 we aim to copy the triple $(F_{k-1}^{(i)}(1), a_i, F_{k-1}^{(i)}(1) \cdot a_i)$ from counters $(o_1, o_2, o_3)$
  back to the counters $(s_1, s_2, s_3)$.
\end{itemize}
The controlling counter $c$ controls counters $o_1$, $o_2$ and $o_3$ for each  $i \in [1,d]$ the test counters of $V_{k-1}^{[i]}$,
which in particular contain counters $s_1$, $s_2$ and $s_3$.
Counter $c$ is designed to test whether each time after line 5 test counters of $V_{k-1}^{[i]}$ are zero
and each time after line 11 counters $o_1$, $o_2$ and $o_3$ are zero.
The first condition assures that indeed output of amplifier $V_{k-1}^{[i]}$ is computed correctly,
while the second condition assures that values of $o_j$ are fully copied back to values of $i_j$.
Notice that each of the test counters of $V_{k-1}^{[i]}$ and $o_j$ counters is tested exactly $n$ times in the program.
In order to fulfil condition (2) of Lemma~\ref{lem:zero-testing} we have the following modifications of counter $\vr{c}$ in our program.
In line 1 we update \add{\vr{c}}{n} and in line 3 we have \add{\vr{c}}{2n}, as counters $s_1$, $s_2$ and $s_3$ are here modified
and all of them await $n$ tests.
In lines 7, 9 and 11 each counter $s_j$ awaits for $n-i$ tests, while counter $o_j$ awaits for $n+1-i$ tests,
this is why counter $c$ is increased here by $-1 = (n-i) \cdot 1 + (n+1-i) \cdot (-1) = -1$.
It is easy to verify conditions (1) and (3) of Lemma~\ref{lem:zero-testing}. Indeed, recall that in order to check whether the considered
program is an $F_k$-amplifier we consider only runs starting with values of all the counters of $V^{[i]}_{k-1}$ beside $(s_1, s_2, s_3)$
being zero (which guarantees condition (1)) and finishing with counter value of controlling counter $c$
being zero (which guarantees condition (3)).
So counter $c$ together with the counters it controls indeed fulfil conditions of the Lemma~\ref{lem:zero-testing}.
Therefore amplifier $V_{k-1}$ computes correctly its output values and values of $o_j$ are correctly transferred to counters $s_j$.
Thus using the induction assumption stating that $V_{k-1}$ is an $F_{k-1}$-amplifier we can easily show that
in the $i$-th iteration of the for-loop we indeed have values of $(o_1, o_2, o_3)$ equal to
$(F_{k-1}^{(i)}(1), a_i, F_{k-1}^{(i)}(1) \cdot a_i)$ for some $a_i \in \N$ guessed nondeterministically.
Therefore after $n$ iterations of the for-loop final values of $(o_1, o_2, o_3)$ are
$(F_{k-1}^{(n)}(1), a_n, F_{k-1}^{(n)}(1) \cdot a_n) = (F_k(n), a_n, F_k(n) \cdot a_n)$ under the condition that
the controlling counter and the input counters are equal to zero at the end of the run. So indeed $V_k$ is an $F_k$-amplifier with
output counters $o_1$, $o_2$ and $o_3$ and test counters $i_1$, $i_2$, $i_3$ and $c$.
\end{proof}

It remains to show how the for-loop and operations \add{\vr{c}}{(n+1)-i} are implemented.
Intuitively speaking it is not problematic because we have an access to the triple of input counters $(i_1, i_2, i_3)$
fulfilling $i_1 i_2 = i_3$ and $n = i_1$. Therefore using the idea of multiplication triples we can allow for zero-testing counters bounded by $n$ and thus also for the needed operations. Below we describe these constructions in detail.

For the implementation of the needed operations we use the input counters $(i_1, i_2, i_3)$ and auxiliary counters $y_1$ and $y_2$.
Assume that the for-loop has the following shape
\begin{algorithmic}[1]
\For {\, $i$ \, := \, $1$ \, \textbf{to } $n$}\label{l:for}
  \State $\langle \text{body} \rangle$
\EndFor
\end{algorithmic}
and inside the $\langle \text{body} \rangle$ we have operations \add{\vr{c}}{(n+1)-i}.
The initial value of $\vr{i_1}$ equals $n$. Below we implement the for-loop
in such a way that all the time counter value of $\vr{i_1}$ is equal to $(n+1)-i$.
So in order to perform \add{\vr{c}}{(n+1)-i} it is enough to implement \add{\vr{c}}{i_1}.
\begin{algorithmic}[1]
\Loop
  \State $\langle \text{body} \rangle$
  \State \sub{\vr{i_1}}{1} \quad \add{\vr{y_2}}{1}
\EndLoop
\end{algorithmic}
As $i_1$ is one of the test counters, we are guaranteed that the loop indeed will be iterated exactly $n$ times.
Now we show how to implement operation \add{\vr{c}}{i_1}, which together shows how to implement
\add{\vr{c}}{n} in lines 1 and 3 and \add{\vr{c}}{(n+1)-i} in the $i$-th iteration of the for-loop.
Operation \sub{\vr{c}}{(n+1)-i} is implemented totally analogously to \add{\vr{c}}{(n+1)-i}.
The implementation works similarly as in the Example~\ref{ex:7VASS}, but here we show how to implement
\add{\vr{c}}{(n+1)-i} using only two, not three, auxiliary counters.
At the beginning we have $i_1 = n$, $y_1 = y_2 = 0$. We will keep the invariant $i_1 + y_1 + y_2 = n$.
Notice that in the $i$-th iteration values of counters are $(i_1, y_1, y_2) = (n+1-i, 0, i-1)$.
We implement the increment \add{\vr{c}}{(n+1)-i} as follows.
\begin{algorithmic}[1]
\Loop
  \State \sub{\vr{i_1}}{1} \quad \add{\vr{y_1}}{1} \quad \add{\vr{c}}{1}
\EndLoop
\State \zerotest{\vr{i_1}}
\Loop
  \State \add{\vr{i_1}}{1} \quad \sub{\vr{y_1}}{1}
\EndLoop
\State \zerotest{\vr{y_1}}
\end{algorithmic}
If zero-tests in lines 3 and 6 are performed correctly then it is easy to see that when
the above program fragment starts in valuation $(i_1, y_1) = (n+1-i, 0)$ it also finishes in the same valuation,
but a side effect is the increment \add{\vr{c}}{i_1}. Thus it remains to show that we can
implement zero-tests or counters $\vr{i_1}$ and $\vr{y_1}$.
We present how to perform \zerotest{\vr{i_1}}, the \zerotest{\vr{y_1}}
is performed analogously with roles of $i_1$ and $y_1$ swapped.
Notice that counter $\vr{i_1}$ is here bounded by $n$, so testing it for zero is much simpler
than testing counters $\vr{s_i}$ or $\vr{o_i}$ above.
On the other hand zero-testing of $\vr{i_1}$ is more complicated than the zero-tests in Example~\ref{ex:7VASS},
as there the tested counters were bounded by the size of VASS transitions. Here the tested
counters are also bounded by $n$, but $n$ is arbitrary, so in order to implement zero-tests
we need to use triples $(i_1, i_2, i_3)$. Additional technical complication is caused by the fact that
we want to optimise the number of the auxiliary counters from three to two.
The \zerotest{\vr{i_1}} is performed as follows.
\begin{algorithmic}[1]
\State \sub{\vr{i_2}}{2}
\Loop
  \State \sub{\vr{y_1}}{1} \quad \add{\vr{i_1}}{1} \quad \sub{\vr{i_3}}{1}
\EndLoop
\Loop
  \State \sub{\vr{y_2}}{1} \quad \add{\vr{y_1}}{1} \quad \sub{\vr{i_3}}{1}
\EndLoop
\Loop
  \State \sub{\vr{y_1}}{1} \quad \add{\vr{y_2}}{1} \quad \sub{\vr{i_3}}{1}
\EndLoop
\Loop
  \State \sub{\vr{i_1}}{1} \quad \add{\vr{y_1}}{1} \quad \sub{\vr{i_3}}{1}
\EndLoop
\end{algorithmic}
We aim to show that if $i_1 + y_1 + y_2 = n$ then the total effect of loops in lines 2-8 on the counter $i_3$
is the decrease by at most $2n$ and the decrease is exactly $2n$ if and only if initially $i_1 = 0$.
As in line 1 counter $i_2$ is decreased by $2$ then counter $i_3$ have to be decreased by exactly $2n$
in the rest of the program fragment, as finally values of both $i_2$ and $i_3$ need to be zero.
Therefore it remains to argue about the decrease of counter $i_3$. The easiest way to see this is to see the loop
in lines 2-3 as transferring value of counter $y_1$ to counter $i_1$, but maybe not fully.
We write it $y_1 \mapsto i_1$.
Similarly next loops correspond to transfers $y_2 \mapsto y_1$, $y_1 \mapsto y_2$ and $i_1 \mapsto y_1$, each of the transfers
may be not fully realised.
The total decrease on $c_3$ equals exactly the total amount of value transferred during all the four loops.
Notice now that the value of original $y_2$ can be used in at most two transfers: $y_2 \mapsto y_1$, $y_1 \mapsto y_2$.
Similarly the value of original $y_1$ can be used either only in $y_1 \mapsto y_2$ or in two transfers
$y_1 \mapsto i_1$ and $i_1 \mapsto y_1$. Value of original $i_1$ can be used only in the transfer $i_1 \mapsto y_2$.
Therefore the total amount of the transfer equals at most $2y_1 + 2y_2 + i_1$ and this equals $2n$ only if $i_1 = 0$.
Moreover in order to obtain the transfer of exactly $2n$ we need to perform all the transfers fully.
Therefore one can easily observe that in that case after the \zerotest{\vr{i_1}} 
values of counters $i_1$, $y_1$ and $y_2$ come back to the same values as before the \zerotest{\vr{i_1}}.
This finishes the proof that the above fragment faithfully implements \zerotest{\vr{i_1}}.

In order to finish the proof of Lemma~\ref{lem:amplifier} it is enough to observe that $\size(V_k)$ is bounded by $C \cdot \size(V_{k-1})$
for some constant $C \in \N$. Indeed $V_k$ was obtained from $V_{k-1}$ by adding several new lines and for each
operation \add{\vr{x}}{1} (or \sub{\vr{x}}{1}) for any controlled counter $\vr{x}$ adding an appropriate operation \add{\vr{c}}{(n+1)-i}
(or \sub{\vr{c}}{(n+1)-i}). Implementing each operation on counter $c$ requires a few new lines which is responsible by multiplying
the size by some constant, but not more. Thus altogether indeed $\size(V_k) \leq C \cdot \size(V_{k-1})$ which finishes the proof.

\section{Future research}\label{sec:future}
We have settled the complexity of the reachability problem for VASSes, but there are still many intriguing questions
in this topic. Here we present some, which we think need investigation in the future works of our community.

We still lack understanding of VASSes in small dimensions. The most striking example is the reachability problem for $3$-VASSes,
where the complexity gap is between \pspace-hardness (inherited from dimension $2$~\cite{DBLP:conf/lics/BlondinFGHM15})
and algorithm working in $\F_7$~\cite{DBLP:conf/lics/LerouxS19}.
We do not see any way of applying our techniques or any other known techniques
of proving lower bounds to dimension $3$, as all of them require some additional counter, which helps to enforce
the run to be exact at some control configurations. We conjecture that the reachability problem is actually elementary
for VASSes in dimension $3$ and maybe even in a few higher dimensions. Showing this seems to be a very challenging task.

Another future goal is to settle the exact complexity of the reachability problem for $d$-VASSes depending on $d$.
After our work showing $\F_d$-hardness in dimension $6d$ and independent work by J{\'{e}}r{\^{o}}me Leroux
~\cite{JeromeAckermann} showing the same in dimension $4d+9$ a lot of improvement appeared very recently.
S{\l}awomir Lasota in~\cite{DBLP:journals/corr/abs-2105.08551} building on results in the arXiv version
of this paper~\cite{DBLP:journals/corr/abs-2104.13866} improved the $\F_d$-hardness result to dimension $3d+2$.
Also J{\'{e}}r{\^{o}}me Leroux improved his own result and
showed $\F_d$-hardness for VASSes in dimension $2d+4$~\cite{DBLP:journals/corr/abs-2104-12695}.
Thus the lower bounds to improve are \tower-hardness for $d \geq 10$~\cite{DBLP:journals/corr/abs-2104-12695},
\expspace-hardness as well for $d \geq 10$~\cite{DBLP:journals/corr/abs-2104-12695}
and $\np$-hardness for $d \geq 7$ for unary encoding~\cite{DBLP:conf/concur/Czerwinski0LLM20}
and the best published upper bounds are stated in~\cite{DBLP:conf/lics/LerouxS19} to be $\F_{d+4}$ complexity for dimension $d$.
One can also suspect that even VASSes in higher dimensions can be solved efficiently
under some condition on their structure (for example avoiding some kind of bad patterns).

Complexity of the reachability problem for VASS extensions such as
pushdown VASSes~\cite{DBLP:journals/lmcs/LerouxPSS19},
branching VASSes~\cite{DBLP:conf/icalp/FigueiraLLMS17}
or data VASSes~\cite{DBLP:conf/fossacs/HofmanLLLST16}
is almost totally unexplored and even decidability is not known for them.
We hope that techniques introduced in this paper may help better understanding the mentioned extensions of VASSes
and in particular prove some complexity lower bounds.

\paragraph*{Acknowledgements}

We thank S{\l}awomir Lasota for many inspiring and fruitful discussions
and Filip Mazowiecki for careful reading of the draft of this paper.
We also thank anonymous reviewers for helpful remarks.

\bibliography{citat}

\end{document}